\documentclass[10pt,conference]{IEEEtran}

\usepackage[T1]{fontenc}
\usepackage{pgfplots}

\usepackage{etex}
\usepackage{psfrag}   
\usepackage{cancel}

\usepackage{url}      
\usetikzlibrary{calc}
\usepackage{pgfplots}
\pgfplotsset{width=7cm,compat=1.3}

\usepackage{psfrag}
\usepackage{makecell}

\usepackage{mathtools}

\usepackage{mathtools}

\usepackage{xfrac}

\usepackage{multicol}

\usepackage {tikz}
\usetikzlibrary {positioning}

\usepackage{latexsym}
\usepackage{pictex}
\usepackage{fancyvrb}
\usepackage{fancyhdr}
\usepackage{color}
\usepackage{subfig}
\usepackage{rawfonts}
\usepackage{graphics}
\usepackage{graphicx}
\usepackage{amssymb,amsmath,amsthm,amsfonts}
\usepackage{bbm}
\usepackage{cite}
\usepackage{enumerate}
\usepackage{verbatim}
\usepackage{balance}

\DeclareMathOperator{\spn}{span}
\newcommand{\nop}[1]{} 
\newcommand{\shorten}[1]{}

\newtheorem{theorem}{Theorem}
\newtheorem{definition}{Definition}

\newtheorem{lemma}{Lemma}

\newtheorem{corollary}{Corollary}

\newcommand{\signed}%
    {{\unskip\nobreak\hfill\penalty50
      \hskip2em\hbox{}\nobreak\hfil $\blacksquare$
      \parfillskip=0pt \finalhyphendemerits=0 \par}}

\begin{document}
\title{Practical Functional Regenerating Codes for Broadcast Repair of Multiple Nodes}

\author{
	\IEEEauthorblockN{Nitish Mital\IEEEauthorrefmark{1}, Katina Kralevska\IEEEauthorrefmark{2}, Cong Ling\IEEEauthorrefmark{1}, and Deniz G\"{u}nd\"{u}z\IEEEauthorrefmark{1} \\ 
	\IEEEauthorblockA{\IEEEauthorrefmark{1}Department of Electrical Electronics Engineering, Imperial College London}
	\IEEEauthorblockA{\IEEEauthorrefmark{2}Dep. of Information Security and Communication Technology, NTNU, Norwegian University of Science and Technology}
	Email: \{n.mital,d.gunduz,c.ling\}@imperial.ac.uk, katinak@ntnu.no}
}

\maketitle

\begin{abstract}
 A code construction and repair scheme for optimal functional regeneration of multiple node failures is presented, which is based on stitching together short MDS codes on carefully chosen sets of points lying on a linearized polynomial. The nodes are connected wirelessly, hence all transmissions by helper nodes during a repair round are available to all the nodes being repaired. The scheme is simple and practical because of low subpacketization, low I/O cost and low computational cost. Achievability of the minimum-bandwidth regenerating (MBR) point, as well as an interior point, on the optimal storage-repair bandwidth tradeoff curve is shown. The subspace properties derived in the paper provide insight into the general properties of functional regenerating codes.
\end{abstract}

\section{Introduction}
The content of a file is typically distributed among multiple access points such that accessing any $k$ distinct access points is sufficient to recover the original file. MDS codes provide high storage efficiency while satisfying the above property. When some nodes fail, their cache contents need to be regenerated to be able to continue serving users. An important objective of edge caching in wireless networks is to reduce the backhaul link loads; therefore, we will consider \textit{cache recovery at the edge}; that is, rather than updating the failed cache contents from a central server through backhaul links, the failed cache contents are regenerated with the help of surviving cache nodes. The total amount of data transferred from the surviving nodes to repair the failed nodes is called the \textit{repair bandwidth}. Traditional MDS codes have high storage efficiency, but their repair bandwidth is large \cite{5550492}. The data of one node is repaired by accessing and transferring data from $k$ nodes, i.e., by recovering the whole content library.

Dimakis et al. showed in  \cite{5550492} that there is a fundamental trade-off between the storage and repair bandwidth by mapping the repair problem in a distributed storage system to a multicasting problem over an information flow graph. The analysis focuses on a single node repair; that is, losing one of the nodes triggers the repair process. Regenerating codes achieve any point on the optimal trade-off curve, while minimum-storage regenerating (MSR) codes and minimum-bandwidth regenerating (MBR) codes operate on the two extremes of this trade-off curve.


It was observed in \cite{5402494} that multiple node repair; that is, the repair process starts only after $r$ nodes fail, is more efficient in terms of the repair bandwidth per node, compared to repairing each node as it fails. In \cite{5978920} and \cite{6565355}, the authors introduce cooperative regenerating codes, which repair multiple failures cooperatively by allowing each of the $r$ nodes being repaired to collect data from $d$ non-failed nodes, called \textit{helper nodes}, and then to cooperate with the other $r-1$ nodes being repaired, called \textit{newcomers}. Cooperative repair allows each newcomer to contact any set of helper nodes independently. An explicit construction of regenerating codes that achieve minimum repair bandwidth under cooperative repair is given in \cite{6566803}.

Another model studied in the literature is the centralized repair model \cite{8469091,7852289}, in which all the helper nodes transmit the repairing symbols to a centralized node, which then repairs the failed nodes. There being a centralized entity repairing the failed nodes, there is no need for the newcomers to exchange data between themselves like in cooperative repair, thus making the system simpler. 

Instead, similarly to \cite{7000553}, we will consider broadcast repair; that is, transmissions from each helper node are received in an error-free manner by all the newcomers. In summary, we will study the broadcast repair of multiple failed cache nodes. The storage-repair bandwidth trade-off for the repair of multiple fully failed nodes is investigated in \cite{7459908},\cite{8613401}.

The broadcast repair model is theoretically equivalent to the centralized multi-node repair model studied in \cite{8469091,7852289}. Hence, the results that we derive, and codes that we construct for broadcast repair are directly applicable to the centralized repair model. The broadcast repair is different from centralized repair in that, while centralized repair involves a centralized entity which then repairs the failed nodes, thus involving two rounds, broadcast repair involves only one round, and thus, is simpler and faster.

In \cite{7852289}, it is shown that the functional MBR point for repair of multiple nodes is not achievable under exact repair. 
Similarly, it is shown that under exact repair, the functional repair tradeoff interior points are also not achievable.  
In \cite{8469091,DBLP:journals/corr/ZorguiW17}, it is shown that cooperative repair achieves the minimum bandwidth of centralized repair (and broadcast repair), under exact repair, albeit at a slightly higher storage cost. 

In this paper, our contribution is to give an explicit code construction to achieve the optimal storage-repair bandwidth tradeoff under functional repair, first for the MBR point for all admissible parameters, and then for an interior point, thus potentially providing us with a general framework to construct functional repair codes to achieve any point on the tradeoff curve. The broadcast nature of the system allows all the newcomers to receive the same data, which simplifies the coding scheme.
Reference \cite{DBLP:journals/corr/abs-1809-08138} studies the functional repair problem from a projective geometry viewpoint. The subspace/projective geometry view makes it hard to visualize how a set of nodes look like, and how one might approach the construction of a code. There have been attempts to derive the conditions necessary for a functional repair code in \cite{DBLP:journals/corr/abs-1809-08138,6620243}. Our code construction strives to address this problem by proposing a simple scheme and simple conditions to guarantee optimality. 

Our code construction has the property that for most failure patterns, the helper nodes do not have to perform computations in a repair round; instead they just read and send the data to the newcomers. This property is called repair-by-transfer, which is desirable for a low I/O cost. Functional repair-by-transfer MDS codes were constructed in \cite{6283043} for some parameters.


Other than these, in our knowledge, there has not been much progress in providing simple, explicit code constructions for functional repair. The network coding literature usually employs random coding, which is not the most practically feasible scheme for distributed storage because of large overheads. 

\section{System Model} \label{nm}
Consider a wireless caching system where $n$ nodes, each with storage capacity $\alpha$ bits, store a file of size $M$ bits. 
We index these storage nodes by the set $\cal N$ $\triangleq \{ 1, \ldots, n \}$. The nodes are fully connected by a wireless broadcast medium and use orthogonal channels for data transmission.

 We refer to the nodes that fail as the \textit{failed nodes} and the nodes that do not experience any losses as the \textit{surviving nodes}. We assume that the repair occurs in rounds, where a repair round gets initiated when $r$ nodes experience failures. Thus, a single repair round repairs $r$ failed nodes. There is no loss during a repair round. During a repair round, the failed nodes are repaired with the help of bits transmitted from $d$ surviving nodes, called the \textit{helper nodes}. 
 
 A data collector (DC) corresponds to a request to reconstruct the file. Data collectors connect to any subset of $k$ active nodes and retrieve all the stored data in these nodes. This is called the \textit{reconstructability property}. In general, the repair is functional, i.e., the repaired content of the node may not be the same as the original content, but it satisfies the reconstructability property. 
 
 \begin{definition} 
The repair bandwidth $\gamma=d\beta$ is defined as the total number of bits the helper nodes broadcast in a repair round.
\end{definition}

\subsection{Subspace view}
Consider a node storing $\alpha$ linearly independent elements $y_1, \ldots, y_{\alpha}$ from $GF(q^m)$ (possible only if $\alpha\leq m$), then any linear operations performed on these finite field elements, which can be viewed as $m-$dimensional vectors over $GF(q)$, lie in the same subspace. Hence, the node is said to store the subspace of dimension $\alpha$, denoted by $W_i \equiv \spn\{y_i\}, i=1,\ldots, \alpha$. For a set of nodes denoted by $\mathcal{A}$, the subspace stored by $\mathcal{A}$ is denoted by $W_{\mathcal{A}} = \sum_{i\in \mathcal{A}} W_i$, where the addition operation on subspaces denotes the direct sum of subspaces. By abuse of notation, $W_i$ can also denote the random variables of the stored information in the nodes. The Shannon entropic measures on random variables can be redefined as a measure $``\dim(.)"$ (subspace dimension) on intersection, union and set difference of subspaces. The following identities hold \cite{79902}:
\begin{align*}
H(W_A)&= \dim(W_A)\\
H(W_1,\ldots,W_l) &= \dim\left( \sum_{i=1}^{l} W_i \right)\\
H(W_A \vert W_B) &= \dim(W_A \setminus W_B)\\
I(W_A;W_B) &= \dim(W_A \cap W_B)\\
I(W_A;W_B;W_C) &= \dim(W_A \cap W_B \cap W_C).
\end{align*}

For each set of parameters, a $(n, k, \gamma,d, \alpha, r)$ tuple is feasible, if a code with storage $\alpha$ and repair bandwidth $\gamma$ exists.

\begin{figure}
    \centering
    \includegraphics[scale=0.3]{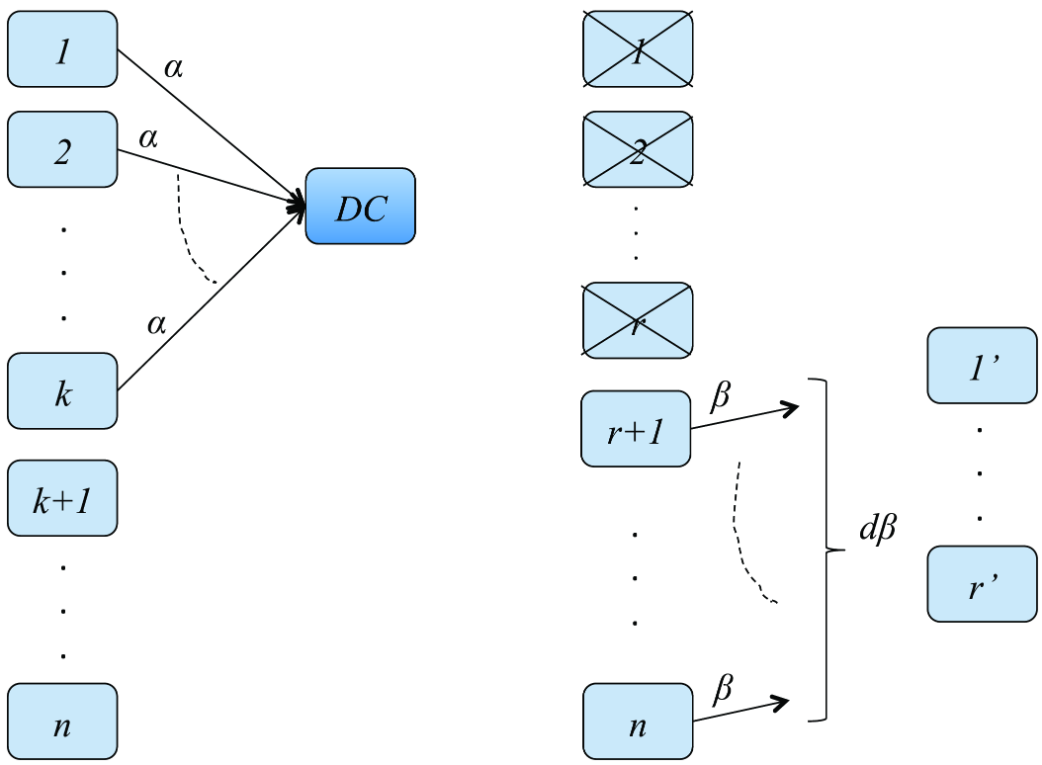}
    \caption{An illustration of the data collection (reconstructability property) and functional repair of $r$ nodes where $d$ nodes are helpers.}
    \label{fig:my_label}
\end{figure}

\section{MBR point construction for multiple node failures}
\begin{theorem} \label{theorem:mbr} 
\cite{7541450} For any $\alpha \geq \alpha^{*}(n,k,\gamma,d, r)$, the points $(n,k, \gamma,d, \alpha, r)$ are feasible, and linear network codes suffice to achieve them. If $r$ divides $k$, the minimum repair bandwidth point is achieved by the pair $(\alpha_{MBR}, \gamma_{MBR}^*)= \frac{2M}{k(2d-k+r)} (d,rd)$.
\end{theorem}

We provide a construction of a functional repair code for the broadcast setting using linearized polynomials, which were first used by Gabidulin for the construction of rank-metric codes \cite{gabidulin1985theory}.
\subsection{Linearized Polynomials}
An important component of our construction is linearized polynomial and their special properties. 

A linearized polynomial 
\begin{align}
    f(x)=\sum_{i=1}^{P}a_i x^{q^{i-1}} ,\ \ \  a_i \in \mathbb{F}_{q^m}
\end{align}
can be uniquely identified from evaluations at any $P$ points $x=\theta_i \in \mathbb{F}_{q^m}, i=1,2,\ldots, P$, that are linearly independent over $\mathbb{F}_q$.

Another relevant property of linearized polynomials is that they satisfy the following condition
\begin{align}
    f(ax+by)=af(x)+bf(y), a,b\in \mathbb{F}_q, x,y \in \mathbb{F}_{q^m},
\end{align}
that is, given a set of points on a linearized polynomial, any linear combination of the points also lies on the polynomial. 

\subsection{Code construction}

Consider a file of $M$ bits. We split the file into $\frac{k}{2}(2d-k+r)$ packets, denoted by $\{m_1, \ldots, m_{\frac{k}{2}(2d-k+r)} \}$. Thus each packet is of size $\frac{2M}{k(2d-k+r)}$ bits. Define the linearized polynomial 
\begin{align}
 f(x)=\sum_{i=1}^{\frac{k}{2}(2d-k+r)}m_i x^{q^{i-1}} ,\ \ \  m_i \in \mathbb{F}_{q^m}
\end{align}
in a field $GF(q^m)$. If a DC receives any $\frac{k}{2}(2d-k+r)$ linearly independent points on the polynomial $f(x)$, it can reconstruct $f(x)$ by interpolation, and thus reconstruct the file.\\

Pick $d$ linearly independent points on $f(x)$, denoted by $(x_1,y_1), \ldots , (x_{d},y_{d})$, where $y_i=f(x_i)$ for all $i=1,\ldots,d$. Store these points at node $1$. Pick another $d$ linearly independent points that are also linearly independent from the points stored at node $1$, denoted by $(x_{d+1},y_{d+1}), \ldots , (x_{2d},y_{2d})$, on $f(x)$, and store them at node $2$, and so on, till $d$ nodes are filled with linearly independent points. For the remaining $n-d$ nodes, fill them as if they are being repaired by the $d$ nodes already filled. The scheme for repair is detailed in the Section \ref{repair}.\\

At each node, encode the points with a $(n-1,d)$ \textit{systematic} MDS code to obtain $n-1$ coded symbols. The symbols at node $i$ are denoted by $\{w_{ij}\}$, where $j\in [n]\setminus \{i \}$.

\subsection{Repair}\label{repair}
Suppose the first repair round repairs the nodes $n-r+1$ to $n$ with the helper nodes $1,\ldots,d$. Each node $i\in [d]$ transmits the points $w_{ij}, j\in [n-r+1:n]$. The total number of points broadcasted by $d$ helper nodes is $rd$. The broadcasted points must be linearly independent, which will become evident that it holds as the repair scheme is explained. Arrange these received points in a $d\times r$ matrix, where row $i$ contains the points transmitted by node $i$. In the following matrix representation of the received points, we denote the points with the node from which it was transmitted only. It must be noted that all elements in the following matrix actually represent distinct points.
\begin{align}\label{unpermuted}
   \mathbf{Y}=\left[ \begin{array}{cccc}
        1 & 1 & \cdots & 1 \\
        2 & 2 & \cdots & 2 \\
        \vdots & \ddots & \ddots & \vdots \\
        d & d & \cdots & d
    \end{array} \right]
\end{align}
We then permute the columns of the above matrix in such a way that no row has two points transmitted by the same node. 
A circular permutation achieving this condition looks like the following
\begin{align} \label{permuted Y}
  \mathbf{Y}= \left[ \begin{array}{cccc}
        1 & 2 & \cdots & r \\ 
        2 & 3 & \cdots & r+1 \\
        \vdots & \ddots & \ddots & \vdots \\
        d & 1 & \cdots & r-1
    \end{array} \right].
\end{align}
Now, assume that each node is given a $(2r,r)$ systematic MDS code generator matrix $\mathbf{G}=[\mathbf{I} \ \ \mathbf{M}]$. $\mathbf{G}$ has rank $r$, implying that any $r$ columns are linearly independent.
Let each newcomer multiply the permuted point-matrix $\mathbf{Y}$ with the local $r \times r$ invertible encoding matrix $\mathbf{M}$ to get $\mathbf{Y}'=\mathbf{YM}$. The points in column $i$ of $\mathbf{Y}'$ are stored on the $i^{th}$ newcomer, and encoded with the local $(n-1,d)$ systematic MDS code.\\
The above scheme ensures the following properties to hold:
\begin{description}
    \item[\textbf{L1:}] For $i\neq j$, $\dim(W_i \cap W_j)=0$, or, $I(W_i ; W_j)=0$. 
      \item[\textbf{L2:}] For any set of nodes $\mathcal{A}$, $\lvert \mathcal{A} \rvert \leq r$, the following holds :
    $\dim(\sum_{i\in \mathcal{A}} W_i)=\sum_{i\in \mathcal{A}} \dim(W_i)$. \\
    This is equivalent to the condition - $H(W_{\mathcal{A}})=\sum_{i\in \mathcal{A}} H(W_i)$.  This property is the consequence of permuting the $\mathbf{Y}$ matrix so that packets from the same node are not repeated in the same row. Encoding each row with the full rank matrix $\mathbf{M}$, extracted from an MDS code, ensures the independence of any $r$ nodes. 
    \item[\textbf{L3:}] Given a node $A$, and a set of nodes denoted by $\mathcal{B}$, $\lvert \mathcal{B} \rvert \leq d$, partition $\mathcal{B}$ into two disjoint subsets $\mathcal{B}_1$ and $\mathcal{B}_2$. Then the following holds - $\dim(W_A \cap W_{\mathcal{B}_1} \cap W_{\mathcal{B}_2})=0$, or, $I(W_A;W_{\mathcal{B}_1};W_{\mathcal{B}_2})=0$. This is because each node transmits a distinct point to repair any node, of which any $d$ of them are linearly independent because of the rank $d$ MDS encoding in each node. Since $d\geq k$, this allows for all admissible parameters.
    \item[\textbf{L4:}] Given a node $A$, and a set of different $r$ nodes denoted by $\mathcal{R}$, then $I(W_A ; W_{\mathcal{R}})\leq r$, where equality holds iff $1)$ the set $\mathcal{R}$ were helper nodes while repairing node $A$, and a particular row of $\mathbf{Y}$ consisted of the nodes in $\mathcal{R}$; or, $2)$ node $A$ was a helper node while repairing the set $\mathcal{R}$ which failed together.
\end{description}
\subsection{Reconstruction} \label{reconstruction}
Suppose a DC accesses the nodes $1,\ldots,k$, denoted by $\mathcal{A}_{dc}$. The points available at the $k$ nodes should be enough to interpolate the polynomial $f(x)$; hence, the necessary and sufficient condition for successful reconstruction is $\dim(W_{\mathcal{A}_{dc}}) \geq \frac{k}{2}(2d-k+r)$. The following lemma will be helpful in showing that the reconstructability property is satisfied.
\begin{lemma}\label{lemma:1}
Consider a node $A$, and a set of other $l\leq d$ nodes. Partition the $l$ nodes into sets of $r$ nodes denoted by $\mathcal{R}_1, \ldots, \mathcal{R}_{\lfloor \sfrac{l}{r} \rfloor} $, and the remaining set of nodes denoted by $\mathcal{R}'$. Then,
\begin{align}
    \dim(W_A \cap \sum_{i=1}^{l} W_i) = \sum_{i=1}^{\sfrac{l}{r}} \dim(W_A \cap W_{\mathcal{R}_i}) 
\end{align}
\begin{proof}
\begin{align*}
    &\dim(W_A \cap \sum_{i=1}^{l} W_i)= I(W_A;W_1,\ldots,W_l)\\
    & = I(W_A;W_{\mathcal{R}_1},\ldots,W_{\mathcal{R}_{\sfrac{l}{r}}}, W_{\mathcal{R}'})\\
    &= I(W_A;W_{\mathcal{R}_1}) + I(W_A;W_{\mathcal{R}_2}\ldots, W_{\mathcal{R}'} \vert W_{\mathcal{R}_1})\\
    & \overset{(a)}{=} I(W_A;W_{\mathcal{R}_1}) + I(W_A;W_{\mathcal{R}_2}\ldots, W_{\mathcal{R}'})
    \end{align*}
    where $(a)$ holds due to $\textbf{L3}$, which applied to a well known identity from multivariate mutual information, gives $I(W_A;W_{\mathcal{R}_2},\ldots, W_{\mathcal{R}'};W_{\mathcal{R}_1})=I(W_A;W_{\mathcal{R}_2},\ldots, W_{\mathcal{R}'}) - I(W_A;W_{\mathcal{R}_2},\ldots, W_{\mathcal{R}'}\vert W_{\mathcal{R}_1})=0$. Using this result inductively, we get
    \begin{align}
        I(W_A;W_{1},\ldots,W_{l})&= \sum_{i=1}^{\sfrac{l}{r}} \nonumber I(W_A;W_{\mathcal{R}_i}) +  I(W_A;W_{\mathcal{R}'}) \nonumber \\
        & \overset{(b)}{=} \sum_{i=1}^{\sfrac{l}{r}}  I(W_A;W_{\mathcal{R}_i}) \label{decomposition} \\
        &= \sum_{i=1}^{\sfrac{l}{r}} \dim(W_A \cap W_{\mathcal{R}_i}) \nonumber
    \end{align}
    where $(b)$ follows from $\textbf{L2}$.
\end{proof}
\end{lemma}
\begin{theorem}
If a DC accesses $k$ nodes, the dimension of the space obtained from those nodes is $\frac{k}{2}(2d-k+r)$. Thus, the DC can reconstruct the file.
\end{theorem}
\begin{proof} Without loss of generality, assume that the DC accesses nodes $1,\ldots, k$. We have
\begin{align*}
    \dim(\sum_{i=1}^{k}W_{i})&= H(W_1,\ldots, W_k)\\
    &=\sum_{i=1}^{k} H(W_i \vert W_{i-1},\ldots,W_1)\\
    &= \sum_{i=1}^{k} \left[ H(W_i) - I(W_i;W_{i-1},\ldots,W_1)\right] \\
    &\overset{(c)}{=} \sum_{i=1}^{k} H(W_i) - \sum_{i=1}^{k} \sum_{j=1}^{\sfrac{(i-1)}{r}} I(W_i;W_{\mathcal{R}_j}) \\
    &\overset{(d)}{\geq} kd- \left(r(r) + r(2r) + \cdots + r(k-r) \right)\\
    &= \frac{k}{2}(2d-k+r)
\end{align*}
where $(c)$ follows from Lemma \ref{lemma:1}, and $(d)$ holds due to $\textbf{L4}$.
\end{proof}

\section{Code construction for interior point, $d=n-r$}
\begin{theorem} \label{theorem3}
\cite{8613401} The pair $(\alpha, \gamma^*)= \frac{2M}{k((2(n-2r)-(k-r))+2r(k-r))} ((n-2r),r(n-r))$ is an interior point on the tradeoff curve, and is achievable with our coding framework.
\end{theorem}
We consider $d=n-r$ in this section. Each file is divided into $\sfrac{1}{2}k((2(n-2r)-(k-r))+2r(k-r))$ subpackets, and the polynomial $f(x)$ is constructed accordingly. The storage and repair scheme is the same as in Section \ref{repair}, except that the matrix $\mathbf{Y}$ is of dimensions $(n-2r)\times 2r$, and $\mathbf{M}$ is a $2r\times r$ local encoding matrix, such that there is a $(3r,2r)$ MDS systematic generator matrix of the form $[\mathbf{I}\ \ \mathbf{M}]$. Each node stores $n-2r$ linearly independent points, which are then encoded with a local systematic $(n-1,n-2r)$ MDS code. 
The matrix $\mathbf{Y}$ is constructed by arranging the points received from $n-2r$ helper nodes in the first $r$ columns, like in Equation \eqref{unpermuted}, and the points received from the remaining $r$ helper nodes in the next $r$ columns. The elements of $\mathbf{Y}$ are rearranged so that no two points from a helper node lie in the same row, similar to Equation \eqref{permuted Y}, to obtain the following form,
\begin{align*}
    \mathbf{Y}= \left[  \begin{array}{c}
      \mathbf{S}_{(n-2r)\times r}  
    \end{array}  
    \begin{array}{|c}
         \mathbf{\phi}_{r\times r}\\ 
        \hline \vdots \\
         \hline \mathbf{\phi}_{r\times r}
    \end{array}
    \right]_{(n-2r)\times 2r}
\end{align*}
where the matrix $\mathbf{\phi}$ consists of the $r^2$ points from the last $r$ helper nodes. Then, the columns of $\mathbf{Y}'=\mathbf{Y}\mathbf{M}$, which is a $(n-2r) \times r$ matrix, are stored on the $r$ nodes respectively. \\

The property $\textbf{L1}$ is satisfied. Properties $\textbf{L2}$,$\textbf{L3}$ and $\mathbf{L4}$ are modified as follows:
\begin{description}
    \item[\textbf{L2:}]  For any set of nodes $\mathcal{A}$, $\lvert \mathcal{A} \rvert \leq 2r$, the following holds :
    $\dim(\sum_{i\in \mathcal{A}} W_i)=\sum_{i\in \mathcal{A}} \dim(W_i)$. \\
    This is equivalent to the condition - $H(W_{\mathcal{A}})=\sum_{i\in \mathcal{A}} H(W_i)$.
    \item[\textbf{L3:}] Given a node $A$, and a set of nodes denoted by $\mathcal{B}$ such that $\lvert \mathcal{B} \rvert \leq n-r$. Partition $\mathcal{B} $ into three disjoint sets $\mathcal{B}_1.\mathcal{B}_2$ and $\mathcal{B}_3$, such that $\lvert \mathcal{B}_3 \rvert =r$, or in other words, $\lvert \mathcal{B}_1 \cup \mathcal{B}_2 \rvert \leq n-2r$. Then the following holds - $I(W_A; W_{\mathcal{B}_1};W_{\mathcal{B}_2}\vert W_{\mathcal{B}_3})=0$. This is a consequence of the fact that each node encodes its stored symbols with a rank $n-2r$ MDS code.
    \item[\textbf{L4:}] Given a node $A$, and a set of different $2r$ nodes denoted by $\mathcal{B}$, the following holds- $I(W_A ; W_{\mathcal{B}}) \leq 2r $.
\end{description}
\textit{Modified Lemma \ref{lemma:1}:} For any node $A$, and a set of different nodes denoted by $ \mathcal{B}$, $\vert \mathcal{B} \vert \leq n-r$, which is partitioned into sets $\mathcal{R}_1,\ldots, \mathcal{R}'$, as in Lemma \ref{lemma:1}, the following holds due to properties $\mathbf{L2,L3}$ and $\mathbf{L4}$:
\begin{align*}
 I(W_A ; W_{\mathcal{B}}) &= \cancel{I(W_A ; W_{\mathcal{R}_1})}+I(W_A ; W_{\mathcal{B}\setminus \mathcal{R}_1}\vert W_{\mathcal{R}_1} )\\
    &\overset{(e)}{=} \sum_{j\geq 2} I(W_A; W_{\mathcal{R}_j}\vert W_{\mathcal{R}_1})
\end{align*}
 where $(e)$ follows similarly to Equation \eqref{decomposition} with the slight difference of conditioning.  Thus, the dimension of the space obtained by a DC accessing any $k$ nodes is 
\begin{align*}
    \dim(\sum_{i=1}^{k}W_{i})
    &= H(W_1,\ldots, W_k)\\
    &= \sum_{i=1}^{k} \left[ H(W_i) - I(W_i;W_{i-1},\ldots,W_1)\right] \\
    &\overset{(f)}{=} \sum_{i=1}^{k} H(W_i) - \sum_{i=1}^{k} \sum_{j=2}^{\sfrac{(i-1)}{r}} I(W_i;W_{\mathcal{R}_j}\vert W_{\mathcal{R}_1}) \\
    &\overset{(g)}{\geq} k(n-2r)- \left( r(r) + \cdots + r(k-2r) \right)\\
    &= \frac{k}{2}\left(2(n-2r)-(k-r)\right) + r(k-r)
\end{align*}
where $(f)$ follows from the modified Lemma 1 above, and $(g)$ holds due to $\mathbf{L4}$.
\subsection{Complexity analysis}
Each node must store a local $r \times r$ encoding matrix, and $n-r$ domain points of the polynomial. The subpacketization level of the scheme for the MBR point is $k(2d-k+r)$. Thus functional MBR codes can have a reasonable subpacketization while achieving the optimal repair bandwidth and low I/O cost. The finite field operations are in $GF(q^m), m\geq d^2 $. This is because while initially storing points in the nodes, the first $d$ nodes must be filled with $d$ linearly independent points, while the content of the remaining nodes can be generated as if they were being repaired by the first $d$ nodes. Addition and subtraction can be done in $\mathbb{O}(m)$, while multiplication and division can be done in $\mathbb{O}(m^{\log_2 3})$ using Karatsuba's algorithm \cite{karatsuba_multiplication_1962}. State of the art polynomial interpolation can be done in $\mathbb{O}(n\log n)$.

The scheme in this paper is partially repair-by-transfer, when the systematic symbols from the stored points are transmitted. When the non-systematic symbols are transmitted, more than one systematic symbol must be read to form the required linear combination. A sparse systematic MDS generator matrix is ideal for this application \cite{746771}.
Another advantage of the scheme is that it stitches together short MDS codes to form a larger code, thus each encoding operation can be done sequentially on a small number of elements.

\begin{figure}
    \centering
    \includegraphics[scale=0.3]{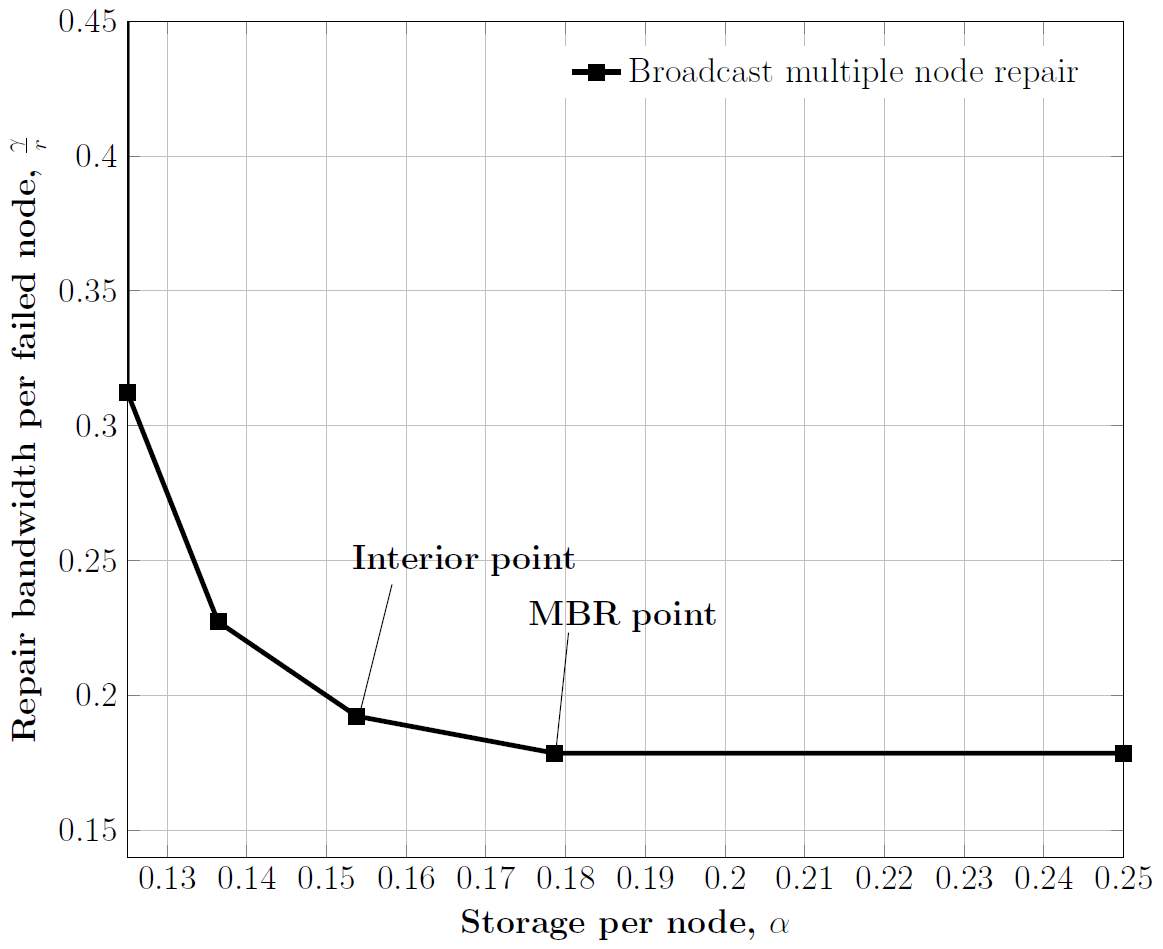}
    \caption{Repair bandwidth per failed node vs. storage per node for $n=12,k=8,d=10,r=2$. The points achieved in this paper are illustrated.}
    \label{fig:plot}
\end{figure}

\section{Conclusions}
We presented a practical code construction and repair scheme for functional repair of multiple node failures in a broadcast setting, achieving the MBR point as well as an interior point, as indicated in Fig. \ref{fig:plot}. We leave for future work to use the construction in this paper for achieving other interior points. The conditions on the stored subspaces provide an insight into the general properties of functional regenerating codes.

\end{document}